\documentclass{article}

\usepackage{fullpage}
\usepackage{amssymb}

\newcommand{\tuple}[1]{\langle #1 \rangle}
\newcommand{\prob}[2]{\mathrm{Prob}(#1\ |\ #2)}
\newcommand{\realworld}[0]{\omega^{\star}}

\usepackage[most]{tcolorbox}
\newtcolorbox{reduction}{breakable,title={Reduction},colback=cyan!5!white}

\newtheorem{defi}{Definition}
\newtheorem{coro}{Corollary}
\newtheorem{lemma}[coro]{Lemma}
\newtheorem{theo}[coro]{Theorem}
\newenvironment{proof}{\itshape}{\hspace{\stretch{1}}QED}

\begin{document}
\title{Complexity Results of Persuasion}
\author{Alban Grastien}
\maketitle

Persuasion is a problem introduced by Wojtowicz \cite{wojtowicz::aies::24}.
It consists in finding a subset of facts that improves
the a-posteriori probability of a specific fact over a given threshold.
It has applications in ``law, advertising, politics, science, public relations, etc''.
We refer the reader to the original paper for a discussion about the benefits of this problem.

The original paper claims that the problem is \textsc{np-complete} and presents a proof.
We claim that that proof was incorrect;
in particular, they make a confusion between the empty conjunction of a collection of sets
and the pairwise empty conjunction.
We prove that persuasion is in \textsc{p} when the persuasion threshold is $1$.
As the reduction proposed by Wojtowicz uses with value for threshold,
this proves the incorrectness of their reduction
(we do not believe that \textsc{p = np} is proved by our simple computations).
We propose a new proof based, as in the original paper, on \textsc{exact cover}
and show \textsc{np-hard}ness in the general case.

\section{Notation and Problem Definition}

% Persuasion is defined as follows.

Given a set $S = \{S_1,\dots,S_k\}$ of sets,
we use $\bigcap S$ to denote $S_1 \cap \dots \cap S_k$
and $\bigcup S$ to denote $S_1 \cup \dots \cup S_k$.

A \emph{probability space} is a triple $P = \tuple{\Omega, F, \pi}$
where $\Omega$ is a finite set of \emph{worlds},
$F \subseteq 2^\Omega$ is a set of \emph{events}
(each event is a subset of worlds),
and $\pi: \Omega \mapsto [0,1]$ is a probability function over $\Omega$.
The natural extension of probability function $\pi$ to sets is $\pi^*$
and is defined by $\pi^*(S) = \sum_{\omega \in S}\ \pi(\omega)$.
It is assumed that the current (unknown) world is $\realworld$,
that it has a non-zero probability ($\pi(\realworld) > 0$),
and that it belongs to the intersection of $F$: $\realworld \in \bigcap F$.
For people who are not familiar with this notation, an event $f$
can be interpreted as a ``fact'';
stating this fact to someone informs the recipient of the statement
that the real world $\realworld$ belongs to the set $f$.

Given a probability space $P = \tuple{\Omega, F, \pi}$,
an \emph{observation} is a subset $R \subseteq F$ of events.
The \emph{probability that an event $e \subseteq \Omega$ is true given an observation $R$}
is calculated by
\begin{equation}
    \prob{e}{R} = \frac{\pi^*(e \cap \bigcap R)}{\pi^*(\bigcap R)}.
    \label{eq::aposteriori}
\end{equation}
Note that the value always sits between $0$ and $1$;
in particular since $\realworld$ belongs to each event $r \in R$,
the denominator is strictly positive  ($\pi^*(\bigcap R) \ge \pi(\realworld) > 0$).

A \emph{persuasion problem instance} (PPI) is a triple $\tuple{P,E,\tau}$
where $P = \tuple{\Omega, F, \pi}$ is a probability space,
$E \subseteq \Omega$ is a particular event called the \emph{goal},
and $\tau \in [0,1]$ is a \emph{threshold}.

A \emph{solution} to PPI $\tuple{P,E,\tau}$ with $P = \tuple{\Omega, F, \pi}$
is an observation $R$ such that $\prob{E}{R} \ge \tau$.

\textsc{Persuasion} is the problem of deciding whether a PPI has a solution.

\paragraph{}It can be useful to reason about events in mirror.
Specifically, what is interesting about a given event $f$
is the set $\Omega \setminus f$ of worlds 
that it does contain.
We say that it \emph{excludes} or \emph{rejects} these worlds.
The worlds rejected by an observation are all the worlds
rejected by at least one event in the observation.

\section{Strong Persuasion}

\begin{defi}
  A \emph{strong persuasion problem instance} (SPPI) is a PPI where the threshold $\tau$ evaluates to~$1$.

  \textsc{Strong Persuasion} is the restriction of \textsc{Persuasion} to SPPIs.
\end{defi}

It is possible to prove that there is a solution to the strong persuasion problem
iff $R = F$ is a solution, 
i.e. $\bigcap F \subseteq E$ holds.

\begin{coro}
  \textsc{Strong Persuasion} is in \textsc{P}.
\end{coro}

If one wanted to drop the condition that the intersection of $F$ is not empty
(but that would be silly),
we could use the following property:
An SPPI admits a solution
iff
there exists $\omega$ in $E$ such that $R_\omega = \{r \in F \mid \omega \in r\}$ is a solution.
Thus, it is possible to go through all elements of $E$
and check in polynomial time whether they yield a solution observation.

\section{Complexity of Persuasion}

Wojtowicz proposes a proof of \textsc{np-hard}ness complexity for \textsc{Persuasion},
but this proof includes an error.
Specifically, they provide a reduction from \textsc{Exact Cover} (defined next subsection)
in which they prove that the intersection of the computed sets is empty;
however, the problem definition requires every pair of set to be disjoint.
We provide a corrected proof here.

\subsection{Exact Cover}

To prove the complexity of \textsc{Persuasion}, we use a reduction from the \textsc{np-complete} problem \textsc{Exact Cover}
which we present now.

An \emph{Exact Cover Instance} (ECI) is a pair $\tuple{S, A}$ where $S = \{s_1,\dots,s_n\}$ is a finite set of objects
and $A = \{A_1,\dots,A_k\}$ is a set of (non-empty, pairwise different) subsets of $S$ such that $\bigcup A = S$ holds.
Elements from $S$ are called \emph{universe elements}
and indices ranging between $1$ and $n$ are \emph{universe indices}.
A \emph{solution to ECI $\tuple{S, A}$} is a subset $H$ of $A$
that forms a partition of $S$, i.e.,
\begin{itemize}
\item $\forall \{A_i,A_j\} \subseteq H.\ A_i \cap A_j \neq \emptyset \Rightarrow A_i = A_j$ and
\item $\bigcup_{A_i \in \mathcal{A}}\ A_i = S$.
\end{itemize}
(A partition should include no empty set, but this is trivially satisfied here
as $A$ has no empty set.)
\textsc{Exact Cover} is the problem of deciding whether there is a solution to an ECI;
it is \textsc{np-complete}.

\subsection{Reduction}
We now propose a reduction from \textsc{Exact Cover} to \textsc{Persuasion}.

We start with a description of the reduction and give the formal definition afterwards.
Our reduction introduces four types of worlds:
\begin{itemize}
\item
  $W_0$ and $X_0$ are two worlds that cannot be rejected by an observation.
  $W_0$ belongs to goal $E$ while $X_0$ does not.
  These two worlds guarantee that we never get any solution with $0$ or $1$ probability.
\item
  $Z$-variables are mapped with universe elements.
  They do not belong to goal $E$.
\item
  $Y$-variables are mapped with pair $\tuple{i,\ell}$ such that $s_\ell \in A_i$.
  They belong to goal $E$.
\end{itemize}
The events are defined to match the sets of $A$.
The event $F_i$ (associated with $A_i$) will exclude precisely all $Y$-variables that mention $i$
and all $Z$-variables mapped with a universe element of $A_i$.

The probabilities of these variables are set in such a way
that the probability of goal $E$ given observation $R$ is attained
iff the following two conditions are satisfied:
\begin{enumerate}
\item 
  All $Z$-variables are excluded from observation $R$.
  This guarantees that the observation matches a cover of $S$.
  Remember that $Z$-variables do not belong to goal $E$;
  thus, giving them a high enough probability will penalise $\prob{E}{R}$
  if they are not rejected.
\item 
  Exactly $n$ $Y$-variables are excluded from observation $R$.
  This guarantees that a given universe element
  is not excluded by two different events from the observation:
  thus the cover is exact.
  As $Y$-variables belong to goal $E$,
  solutions are naturally better when they do not exclude them;
  yet their probability is set to a low value
  so that this objective is superseeded by the first one
  (which guarantees the solution is a cover).
\end{enumerate}

Our proposed reduction is presented on Fig.~\ref{fig::reduction}.

\begin{figure}[t!]
\begin{reduction}
Let $\tuple{S,A}$ be an ECI where $S = \{s_1,\dots,s_n\}$ and $A = \{A_1,\dots,A_k\}$.
Let $m = \sum_{i \in \{1,\dots,k\}}\ |A_i|$ be the number of elements in the subsets of $A$.
We define the PPI $\tuple{P,E,\tau}$ with $P = \tuple{\Omega, F, \pi}$ as follows:
\begin{itemize}
  \item $\Omega = W \cup X \cup Y \cup Z$ where
  \begin{itemize}
    \item $W = \{ W_0 \}$ is a singleton,
    \item $X = \{ X_0 \}$ is a singleton,
    \item $Y = \{ Y_{i,\ell} \mid i \in \{1,\dots,k\}, \ell \in \{1,\dots,n\} \textnormal{ such that } s_\ell \in A_i \}$ contains $m$ variables, and
    \item $Z = \{ Z_\ell \mid \ell \in \{1,\dots,n\} \}$ contains $n$ variables,
  \end{itemize}
  \item $F = \{F_1,\dots,F_k\}$ where $\forall i \in \{1,\dots,k\}$
  \begin{displaymath}
    F_i = \Omega \setminus \{Y_{i,\ell} \mid s_\ell \in A_i\} \setminus \{Z_{\ell} \mid s_\ell \in A_i\},
  \end{displaymath}
  \item $\pi$ satisfies 
  \begin{displaymath}    
  \pi = \left\{
    \begin{array}{ll}
    \pi(W_0) = \pi(X_0) = x, &\\
    \pi(Y_{i,\ell}) = y & i \in \{1,\dots,k\}, \ell \in \{1,\dots,n\} \textnormal{ such that } s_\ell \in A_i \\
    \pi(Z_\ell) = z & \ell \in \{1,\dots,n\},
    \end{array}
  \right.
\end{displaymath}
where
\begin{itemize}
  \item $x = \frac{1}{3}$,
  \item $y = \frac{1-2x}{m(1+2n)}$,
  \item $z = 2my %\frac{1-2x-my}{k}
  $, and
\end{itemize}
\item $E = W \cup Y$, and
\item $\tau = \frac{x + (m-n)y}{2x + (m-n)y}$.
\end{itemize}

Given a set $R \subseteq F$ of facts, we define the set $H_R \subseteq A$ of subsets by
\begin{displaymath}
  H_R = \{ A_i \mid \exists r \in R.\ \exists \ell \in \{1,\dots,n\}.\ s_\ell \in A_i \land Y_{i,\ell} \not\in r \}
\end{displaymath}
\end{reduction}
\caption{Reduction from \textsc{Exact Cover} to \textsc{Persuasion}.}
\label{fig::reduction}
\end{figure}

%We call \emph{$Y$-variable} any variable $Y_{i,\ell}$ for some $i$ and $\ell$.
%Similarly, a \emph{$Z$-variable} is a variable $Z_i$ for some $i$.
%In the following work, we will ignore the trivial cases where $m < n$
%(i.e, when the sets of $A$ do not cover $S$).

\subsection{Properties of the reduction}

\begin{lemma}
  $\pi$ is a probability function over $\Omega$.
\end{lemma}

\begin{proof}
  We need to verify two properties: that the probability for each element in $\Omega$ is between $0$ and $1$,
  and that the probabilities add up to $1$.
  Remember that there are $m$ $Y$-variables and $n$ $Z$-variables. 

  \begin{enumerate}
    \item 
  Probabilities of $W_0$ and $X_0$ equal $1/3$ and are thus in $[0,1]$.
  Probability of a $Y$-variable is $(1-2x) / (m(1+2n)) = (1/3) / (m(1+2n))$;
  thus it is within $[0,1]$.
  Probability of a $Z$-variable is $z = 2my = 2m(1-2x)/(m(1+2n)) = 2(1-2x)/(1+2n) = (2/3)/(1+2n)$;
  thus it is within $[0,1]$.

  \item 
  The sum of probabilities is
  \begin{displaymath}
  \begin{split}
    \sum_{\omega \in \Omega} \pi(\omega)
    & = x+x+my+nz \\
    & = x+x + my + 2nmy \\
    & = x+x + (1+2n)my \\
    & = x+x + (1+2n)m\frac{1-2x}{m(1+2n)} \\
    & = x+x + (1-2x) \\ & = 1.
  \end{split}
  \end{displaymath}
  \end{enumerate}
\end{proof}

\begin{lemma}\label{lemma::0zvar}
  Any $R$ such that $\bigcap R$ contains one of more $Z$-variable satisfies $\prob{E}{R} < \tau$. 
\end{lemma}

\begin{proof}
  Let us assume that $\bigcap R$ contains at least one $Z$-variable.
  Let $m'$ and $n'$ be the number of $Y$-variables and $Z$-variables that appear
  in $\bigcap R$.
  Notice: $m' \in \{0,\dots,m\}$ and $n' \in \{1,\dots,n\}$.

  \begin{displaymath}
    \begin{split}
      \prob{E}{R} - \tau 
      & = \frac{x+m'y}{2x+m'y+n'z} - \tau\\
      & < \frac{x+my}{2x+z} - \tau \\
      & = \frac{x+my}{2x+z} - \frac{x + (m-n)y}{2x + (m-n)y}\\
      & \le \frac{x+my}{2x+z} - \frac{x}{2x}\\
      & = \frac{2x(x+my) - x(2x+z)}{2x(2x+z)}\\
      & = \frac{1}{2(2x+z)} (2(x+my) - (2x+z))\\
      & = \frac{1}{2(2x+z)} (2x+2my - 2x-z)\\
      & = \frac{1}{2(2x+z)} (2my -z)\\
      & = 0.
    \end{split}
  \end{displaymath}
  Thus, $\prob{E}{R}$ is strictly less than $\tau$.
\end{proof}

\begin{coro}\label{coro:yvars<}
  Any $R$ such that $\bigcap R$ contains strictly fewer than $(m-n)$ $Y$-variables satisfies $\prob{E}{R} < \tau$. 
\end{coro}

\begin{proof}
  Let $m'$ and $n'$ be the number of $Y$-variables and $Z$-variables that appear
  in $\bigcap R$.
  Assume $\prob{E}{R} = \tau$.
  From Lemma~\ref{lemma::0zvar}, we know $n'=0$.
  Thus, 
  \begin{displaymath}
    \begin{split}
      0 
      & = \prob{E}{R} - \tau \\
      & = \frac{x+m'y}{2x+m'y+n'z} - \tau \\
      & = \frac{x+m'y}{2x+m'y} - \tau\\
      & = \frac{x+m'y}{2x+m'y} - \frac{x + (m-n)y}{2x + (m-n)y},
    \end{split}
  \end{displaymath}
  Thus, $m' = m-n$.
\end{proof}

\begin{coro}\label{coro:yvars>}
  Any $R$ that satisfies $\prob{E}{R} \ge \tau$
  contains exactly $(m-n)$ $Y$-variables.
\end{coro}

\begin{proof}
  Assume $R$ is a solution to the PPI.
  From Lemma~\ref{lemma::0zvar}, we know that $\bigcap R$ contains no $Z$-variables,
  i.e., for each $\ell \in \{1,\dots,n\}$. there is $i \in \{1,\dots,k\}$,
  such that hold both $Z_i \not\in F_\ell$ and $F_\ell \in R$.
  But, by definition of the reduction,
  we also have $Y_{i,\ell} \not\in F_i$ and $Y_{i,\ell} \not\in \bigcap R$.
  In other words, for each $\ell \in \{1,\dots,n\}$,
  there is a different $Y$-variable not in $\bigcap R$,
  i.e., the number of $Y$-variables in $\bigcap R$ is at most $(m-n)$.
\end{proof}

%PREVIOUS LEMMA / COROLLARIES NEED TO PROVE:
%$\forall i. \exists! \ell.\ Y_{i,\ell} \in \Omega \setminus \bigcap R$

\begin{coro}\label{coro::exactcover}
  If $R$ is a solution to the PPI, then $H_R$ is a solution to the ECI.
\end{coro}

% WE ARE ASSUMING THAT NO $A_i$ IS EMPTY.

\begin{proof}
  Assume $R$ is a solution to the PPI.

  Firstly, let $\ell \in \{1,\dots,n\}$ be a universe index;
  from Lemma~\ref{lemma::0zvar}, we know that $\bigcap R$ does not contain $Z_\ell$. 
  Thus, looking at the subsets of $F$ that exclude $Z_\ell$,
  $R$ must contain a set $F_i$ ($i \in \{1,\dots,k\}$)
  such that $s_\ell \in A_i$.
  This $F_i$ also excludes $Y_{i,\ell}$.
  Thus, by definition of $H_R$, this $A_i$ belongs to $H_R$.
  Therefore, $\bigcup H_R$ includes $s_\ell$, and all elements in $S$.

  Secondly, 
  we just mentioned that for each index $\ell$, there is at least one variable $Y_{i,\ell}$ not in $\bigcap R$.
  Assume, by contradiction, that $H_R$ contains two different sets $A_i$ and $A_j$ that are not disjoint, 
  i.e., such that some universe index $\ell$ belongs to $A_i \cap A_j$.
  Then, both $Y_{i,\ell}$ and $Y_{j,\ell}$ are excluded by $\bigcap R$.
  However, this implies that $\bigcap R$ excludes at least $n+1$ $Y$-variables.
  This contradicts Corollary~\ref{coro:yvars>} which states 
  that $\bigcap R$ contains exactly $(m-n)$ $Y$-variables.
  Therefore all sets in $H_R$ are disjoint.

  In conclusion, $H_R$ is an exact cover of $S$.
\end{proof}

\begin{theo}
  \textsc{persuasion} is \textsc{np-hard}.
\end{theo}

\begin{proof}
  This is a consequence of Corollary~\ref{coro::exactcover}
  which showed that our reduction from \textsc{exact cover} to \textsc{persuasion} is correct.
\end{proof}

%\paragraph{Complexity Proof}

%Here, we assume that an ECI $\tuple{S, A}$
%has been reduced to a PPI $\tuple{P,E,\tau}$ where $P = \tuple{\Omega, F, \pi}$,
%and we shall not repeat this in the formal results.
%
%Given a subset $R \subseteq F$ of events,
%we use the notation $P_R = |P \cap \bigcap R|$ and $Q_R = |Q \cap \bigcap R|$.
%We note that the probability of $E$ given $R$ is
%\begin{displaymath}
%  \prob{E}{R} = \frac{\epsilon P_R}{.5 + c Q_R}.
%\end{displaymath}

%\begin{lemma}
%  A solution $R$ to PPI $\tuple{P,E,\tau}$ satisfies the property $Q_R = 0$.
%\end{lemma}
%\begin{proof}
%  Assume $Q_R \ge 1$.
%  Then we have:
%  \begin{displaymath}
%    \begin{array}{rcl}
%    \prob{E}{R} 
%    & = & \frac{\epsilon P_R}{.5 + c Q_R} \\
%    & \le & \frac{\epsilon P_R}{.5 + c} \\
%    & \le & \frac{\epsilon |P|}{.5 + c} \\
%    & = & \frac{\frac{.5 - |Q| c}{|P|} |P|}{.5 + c} \\
%    & = & \frac{.5 - |Q| c}{.5 + c} \\
%    & = & \frac{.5 - |Q| .5 \frac{1+|Q|}{|P|-|Q|}}{.5 + .5\frac{1+|Q|}{|P|-|Q|}} \\
%    & = & \frac{1 - |Q| \frac{1+|Q|}{|P|-|Q|}}{1 + \frac{1+|Q|}{|P|-|Q|}} \\
%    & = & \frac{\frac{|P|-|Q| - |Q| (1+|Q|)}{|P|-|Q|}}{\frac{|P|-|Q|+1+|Q|}{|P|-|Q|}} \\
%    & = & \frac{|P|-|Q| - |Q| (1+|Q|)}{|P|-|Q|+1+|Q|} \\
%    & = & \frac{|P|-2|Q|-|Q|^2}{|P|+1} \\
%    \end{array}
%  \end{displaymath}
%\end{proof}

\bibliographystyle{alpha}
\bibliography{main.bib}

\end{document}